\documentclass[letter,11pt]{article}

\usepackage[bottom=1in,top=1in,left=1in,right=1in]{geometry}
\usepackage{tikz}

\usepackage{writes}
\usepackage{mathtools}

\newcommand{\RR}{\mathbb{R}}
\newcommand{\Ex}[2][]{\mbox{\rm\bf E}_{#1}\hspace{-0.03cm}\left[#2\right]}
\newcommand{\calE}{\mathcal{E}}
\newcommand{\calF}{\mathcal{F}}
\renewcommand{\Pr}{\mbox{\rm\bf Pr}}
\usepackage{thmtools}
\usepackage{thm-restate}

\usepackage{url,hyperref}

\allowdisplaybreaks

\title{Noisy, Greedy and Not So Greedy $k$-means++}
\author{Anup Bhattacharya\thanks{Indian Statistical Institute, India} \and Jan Eube\thanks{University of Bonn, Germany} \and Heiko Röglin\footnotemark[2] \and Melanie Schmidt\footnotemark[2]}
\date{}

\begin{document}

\maketitle

\begin{abstract}
The $k$-means++ algorithm due to Arthur and Vassilvitskii~\cite{AV07} has become the most popular seeding method for Lloyd's algorithm. It samples the first center uniformly at random from the data set and the other $k-1$ centers iteratively according to~$D^2$-sampling, i.e., the probability that a data point becomes the next center is proportional to its squared distance to the closest center chosen so far. $k$-means++ is known to achieve an approximation factor of $\mathcal{O}(\log k)$ in expectation.

Already in the original paper on $k$-means++, Arthur and Vassilvitskii suggested a variation called \emph{greedy $k$-means++ algorithm} in which in each iteration multiple possible centers are sampled according to~$D^2$-sampling and only the one that decreases the objective the most is chosen as a center for that iteration. It is stated as an open question whether this also leads to an $\mathcal{O}(\log k)$-approximation (or even better). We show that this is not the case by presenting a family of instances on which greedy $k$-means++ yields only an $\Omega(\ell\cdot\log k)$-approximation in expectation where $\ell$ is the number of possible centers that are sampled in each iteration.

We also study a variation, which we call \emph{noisy $k$-means++ algorithm}. In this variation only one center is sampled in every iteration but not exactly by $D^2$-sampling anymore. Instead in each iteration an adversary is allowed to change the probabilities arising from $D^2$-sampling individually for each point by a factor between~$1-\eps_1$ and~$1+\eps_2$ for parameters~$\eps_1\in[0,1)$ and $\eps_2\ge 0$. We prove that noisy $k$-means++ computes an $\mathcal{O}(\log^2 k)$-approximation in expectation. We also discuss some applications of this result.
\end{abstract}

\thispagestyle{empty}
\newpage
\setcounter{page}{1}


\section{Introduction}

Clustering is a very important tool in many machine learning applications. The task is to find structure that is hidden in input data in the form of clusters, and to do this in an unsupervised way. Since clusters come with very different properties depending on the application, a variety of clustering algorithms and measures to judge clusterings have arisen in the last decades. Among those, a hugely popular method is Lloyd's algorithm~\cite{L57} (also called the $k$-means algorithm), which for example was voted to be one of the ten most influential data mining algorithms in machine learning at
 the IEEE International Conference on Data Mining (ICDM) in 2006~\cite{WK+08}.

Lloyd's algorithm is an iterative local search heuristic operating on points from Euclidean space~$\mathbb{R}^d$. The measure that it implicitly strives to optimize is the $k$-means cost function: For a point set $X \subset \mathbb{R}^d$ and a center set $C \subset \mathbb{R}^d$, the $k$-means cost function is defined as
\[
\Phi(X,C) = 
\sum_{x \in X} \min_{c\in C} ||x-c||^2,
\]
the sum of the squared distances of all points to their respective center. The $k$-means problem asks to minimize the $k$-means cost over all choices of $C$ with $|C|=k$. In an optimal solution of the $k$-means problem, the centers are means of their clusters, and Lloyd's algorithm iterates between computing the means of all clusters as the new center set and reassigning all points to their closest centers to form new clusters. The $k$-means cost function is also called \emph{sum of squared errors} because when the means are viewed as representatives of the clusters, then the $k$-means cost is the squared error of this representation. 

The $k$-means problem is NP-hard~\cite{ADHP09,MNV09}, and it is also hard to approximate to arbitrary precision~\cite{ACKS15,LSW17}. On the positive side, constant-factor approximations are possible, and the best known factor is 6.357 due to a break-through result by Ahmadian et al.~\cite{ANSW17,LSW17}. However, the constant-factor approximation algorithms for $k$-means are not very practical. On the other hand, Lloyd's method is hugely popular in practice, but can produce solutions that are arbitrarily bad in the worst case.

\begin{algorithm}[b]
	\caption{The $k$-means++ algorithm~\cite{AV07}}\label{alg:kmeans++}
	\begin{algorithmic}[1]
		\State Sample a point $c_1$ independently and uniformly at random from $X$.
		\State Let $C=\{c_1\}$.
		\For{$i=2$ to $k$}
		  \State For all $x \in X$, set \[ p(x) := \frac{\min_{c \in C} ||x-c||^2}{\sum_{y\in X} \min_{c \in C} ||y-c||^2}. \]
		  \State Sample a point $c_i$ from $X$, where every $x \in X$ has probability $p(x)$.
			\State Update $C=C\cup \{c_i\}$.
		\EndFor
		\State Run Lloyd's algorithm initialized with center set $C$ and output the result.
	\end{algorithmic}
\end{algorithm}
A major result in clustering thus was the $k$-means++ algorithm due to Arthur and Vassilvitskii~\cite{AV07} in 2007, which enhances Lloyd's method with a fast and elegant initialization method that provides an $\mathcal{O}(\log k)$-approximation in expectation. 
The $k$-means++ algorithm samples $k$ initial centers by adaptive sampling, where in each step, a point's probability of being sampled is proportional to its cost in the current solution (we will refer to this kind of sampling as $D^2$-sampling in the following). After sampling $k$ centers, the solution is refined by using Lloyd's algorithm. Algorithm~\ref{alg:kmeans++} contains pseudo code for the $k$-means++ algorithm.

The beauty of the algorithm is that it has a bounded approximation ratio of $\mathcal{O}(\log k)$ in expectation, and at the same time computes solutions that are good (much better than $\Theta(\log k)$) on practical tests. By feeding the computed centers into Lloyd's method, the solutions are refined to even better quality. 
Nevertheless, Arthur and Vassilvitskii  show that the approximation ratio of $k$-means++ is tight in the worst case: They give an (albeit artificial) example where the expected approximation ratio is $\Omega(\log k)$, and this has been extended by now to examples where $k$-means++ outputs a $\Omega(\log k)$-approximate solution with high probability~\cite{BR13}, and even in the plane~\cite{BJA16}.

Due to its beneficial theoretical and practical properties, $k$-means++ has by now become the de-facto standard for solving the $k$-means problem in practice. What is less known is that the original paper~\cite{AV07} and the associated PhD thesis~\cite{V07} actually propose a possible improvement to the $k$-means++ algorithm: the \emph{greedy $k$-means++ algorithm}. Here in each of the adaptive sampling steps, not only one center but $\ell$ possible centers are chosen (independently according to the same probability distribution), and then among these the one is chosen that decreases the $k$-means costs the most. 
This is greedy because a center that reduces the cost in the current step might be a bad center later on (for example if we choose a center that lies between two optimum clusters, thus preventing us from choosing good centers for both on the long run). The original paper~\cite{AV07} says:

\begin{quote}
\emph{Also, experiments showed that k-means++ generally performed better if it selected several new centers during
each iteration, and then greedily chose the one that decreased $\Phi$ (the cost function) as much as possible. Unfortunately, our proofs do not carry over to this scenario. It would be interesting to see a comparable (or better) asymptotic result proven here.}
\end{quote}

The intuition is that $k$-means++ tries to find clusters in the dataset, and with each sample, it tries to find a new cluster that has not been hit by a previously sampled center. This has a failure probability, and the super-constant approximation ratio stems from the probability that some clusters are missed. In this failure event, the algorithm chooses two centers that are close to each other compared to the optimum cost. Greedy $k$-means++ tries to make this failure event less likely by boosting the probability to find a center from a new cluster that has not been hit previously and greedily choosing the center.

For $\ell = 1$, the greedy $k$-means++ becomes the $k$-means++ algorithm, and for very large $\ell$ it becomes nearly deterministic, a heuristic that always chooses the current minimizer among the whole dataset. It is easy to observe that the latter is not a good algorithm: Consider Figure~\ref{fig:det:badexample}. In the first step, the center that minimizes the overall $k$-means cost in the next step is $b$. But if we choose $b$, then the second greedy center is either $a$ or $c$, and we end up with a clustering of cost $\Omega(n)$, while the solution $\{a,c\}$ has a cost of $1$ (and the optimum solution is even slightly better). 
\begin{figure}
\centering
\begin{tikzpicture}
\node (a) [circle,draw,fill,inner sep=0cm,minimum width=0.2cm,label=above:{$n$},label=below:{$a$}] at (-4,0) {};
\node (b) [circle,draw,fill,inner sep=0cm,minimum width=0.2cm,label=above:{$1$},label=below:{$b$}] at (0,0) {};
\node (c) [circle,draw,fill,inner sep=0cm,minimum width=0.2cm,label=above:{$n$},label=below:{$c$}] at (4,0) {};
\draw [<->] (a) to node [fill=white] (l) {$1$} (b);
\draw [<->] (b) to node [fill=white] (l) {$1$} (c);
\end{tikzpicture}
\caption{A bad example for the deterministic heuristic that always chooses the current cost minimizer as the next center. An optimal $2$-clustering costs less than $1$, while a clustering where $b$ is a center costs $\Omega(n)$. \label{fig:det:badexample}}
\end{figure}
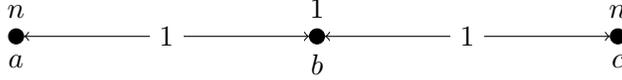

So the crucial question is how to set $\ell$, and whether there is an $\ell$ for which greedy $k$-means++ outperforms $k$-means++. 
It has been shown in~\cite{ADK09} that for any optimal clustering of an input data set, $k$-means++ has a constant probability to sample a point from a \lq new\rq{} optimal cluster in each iteration, where new means that no point from that cluster has previously been chosen as a center. This leads to a bicriteria approximation, since after $\mathcal{O}(k)$ centers, the algorithm has discovered all optimal clusters in expectation.
Following the intuition that stems from this analysis, a natural idea would be to set $\ell = \mathcal{O}(\log k)$: This reduces the probability to pick no point from a new cluster to $\Omega(1/k)$, and by union bound, the failure probability that this event happens in one of the $k$ samples decreases to a constant. 
This choice is also advertized by Celebi et al.~\cite{CKV13}, who feature greedy $k$-means++ in a study of initialization strategies for Lloyd's method. They report that it performs better than $k$-means++, for a suggested value of $\ell = \log k$. The PhD thesis~\cite{V07} reports experiments with $\ell = 2$ that outperformed $k$-means++. It also states that the approximation guarantee of greedy $k$-means++ is unknown (pp. 62+63). 

We initiate the analysis of the greedy $k$-means++ algorithm.
Firstly, we prove that greedy $k$-means++ is \emph{not} asymptotically better than $k$-means++. More precisely, we show the following statement.

\begin{restatable}{theorem}{thmLowerBound}\label{thmLowerBound}
For any $k \ge 4$ and any~$\ell$, there exists a point set $X_{k,\ell}$ such that the expected approximation guarantee of greedy $k$-means++ is~$\Omega(\min\{\ell,k/\log{k}\}\cdot \log k)$.
\end{restatable}

Theorem~\ref{thmLowerBound} implies that the worst-case approximation guarantee of greedy $k$-means++ cannot get better by choosing $\ell > 1$. In particular for $\ell = \log k$, the approximation guarantee worsens to $\Omega(\log^2 k)$. The lower bound example is close to the original lower bound example in \cite{AV07}, yet the proof of the lower bound proceeds very differently. Morally, instead of missing clusters (which becomes less likely due to the multiple samples), the failure event is to choose a bad point as a center. This alone is responsible for the $\Omega(\ell \log k)$ lower bound, while the original $\Omega(\log k)$ bound stems from missing clusters.

As indicated in the quote from $\cite{AV07}$ above, the original proof of $k$-means++ does not carry over to greedy $k$-means++, not even if we aim for a higher approximation guarantee like $\mathcal{O}(\ell \log k)$. Roughly speaking, the main problem in the analysis is that while the probability to choose a point as a center can only be increased by a factor of $\ell$ by the greedy procedure, there is no multiplicative lower bound on how much individual probabilities can be \emph{de}creased. Indeed, if a point $x\in P$ is the worst greedy choice, then its probability to be chosen decreases from some $p(x)$ in the original $k$-means++ algorithm to $(p(x))^{\ell}$, which is much smaller than $p(x)$. If this happened to potentially good centers, then it could  hurt the approximation factor badly.

We proceed to study a different variation of $k$-means++, which may be of independent interest, the \emph{noisy $k$-means++} algorithm. This algorithm performs $k$-means++, but does not sample with exact probabilities. Instead of sampling a point $x$ with probability $p(x)$ as suggested by $D^2$-sampling, it uses an arbitrary probability $p'(x)$ with $(1-\epsilon_1)p(x) \le p'(x) \le (1+\epsilon_2) p(x)$, where $\epsilon_1 \in [0,1)$ and $\epsilon_2 \ge 0$.
 If we cast greedy $k$-means++ as a noisy $k$-means++ algorithm, we observe that we get a trivial upper bound of $\epsilon_2 = \ell-1$, however, no trivial lower bound on how much the probabilities are skewed.

Noisy $k$-means++ is also interesting it its own right, since in practice, the probabilities actually computed are prone to rounding errors. Due to the iterative nature of $k$-means++, it is not at all clear how large the effect of a small rounding can be. We show that the following theorem holds.

\begin{restatable}{theorem}{thmNoisy}\label{thm:noisy}
Let $T_k$ denote the set of centers sampled by noisy $k$-means++ on dataset $X$ and assume that $\frac{k}{\ln k} \ge \max\{18,\frac{24(\eps_1+\eps_2)(1+\eps_2)}{(1-\eps_1)^2}\}$. Then, 
\[
  \E[\Phi(X,T_k)]\leq \mathcal{O}\left(\left(\frac{1+\eps_2}{1-\eps_1}\right)^3\cdot \log^2(k)\cdot \OPT_k(X)\right). 
\]
If $\frac{k}{\ln k} \le \max\{18,\frac{24(\eps_1+\eps_2)(1+\eps_2)}{(1-\eps_1)^2}\}$, 
then $\E[\Phi(X,T_k)]\leq \mathcal{O}\left(\left(\frac{1+\eps_2}{1-\eps_1}\right)^4\cdot \log^2 \left(\frac{1+\eps_2}{1-\eps_1}\right) \cdot \OPT_k(X)\right)$.
\end{restatable}

We use Theorem~\ref{thm:noisy} to analyze a \emph{moderately greedy} variant of $k$-means++, where the simple idea is that with probability $p$, we do a normal $k$-means++ step, and with probability $1-p$, we do a greedy $k$-means++ step. 
The idea is that in this variant, a point is never completely disregarded, so we do get a lower bound on the probabilities, yet in many steps, we do still profit from the additional power of greedy $k$-means++ seen in experiments. For constant $p$ and $\ell$, this variant gives an $\mathcal{O}(\log^2 k)$-approximation by Theorem~\ref{thm:noisy}.

\paragraph{Additional related work.}
Bachem et al.~\cite{BachemLHK16} suggest to speed up $k$-means++ by replacing the exact sampling according to the probabilities~$p(x)$ by a fast approximation based on Markov Chain Monte Carlo sampling. They prove that under certain assumptions on the dataset their algorithm yields the same approximation guarantee in expectation as $k$-means++, namely~$\mathcal{O}(\log{k})$. Their algorithm can be viewed as a special case of noisy $k$-means++. However, their analysis of the approximation factor is based on making the total variation distance between the probability distributions~$p$ and~$p'$ (in every step) so small that with high probability their algorithm behaves identically to $k$-means++. In contrast to this, Theorem~\ref{thm:noisy} also applies to choices of~$\eps_1$ and~$\eps_2$ for which noisy $k$-means++ behaves differently from $k$-means++ with high probability.

Lattanzi and Sohler~\cite{LS19} propose an intermediate improvement step to be executed between the $D^2$-sampling and Lloyd's algorithm in order to improve the solution quality to a constant factor approximation in expectation. Their algorithm starts with a $k$-means++ solution and then performs $\mathcal{O}(k \log\log k)$ improvement steps: In each such step, a new center is sampled with $D^2$-sampling, and if swapping it with an existing center improves the solution, then this swap is performed. While this is a greedy improvement step and thus a bit related to greedy $k$-means++, their algorithm is closer in spirit to a known local search algorithm by Kanungo et al.~\cite{KMNPSW04} which uses center swaps (starting on an arbitrary solution) to obtain a constant-factor approximation, but needs a lot more rounds and is impractical.

In his master's thesis, Pago~\cite{P18} shows that for $\ell = \log k$, the example in Figure~\ref{fig:det:badexample} can be extended such that greedy $k$-means++ gives an $\Omega(\log k)$-approximation in expectation.

The bicriteria analysis by Aggarwal et al.~\cite{ADK09} mentioned above was improved by Wei~\cite{W16} who showed that for any $\beta > 1$, sampling $\beta k$ centers with $D^2$-sampling yields an $\mathcal{O}(1)$-approximation in expectation (with $\beta k$ centers). All above cited works assume that $k$ and $d$ are input parameters; if one of them is a constant, then there exists a PTAS for the problem~\cite{CAKM16,FL11,FRS19}.


\section{Lower Bound for Greedy \texorpdfstring{$k$}{k}-means++}\label{sec:strong-lower}

In this section we construct an instance on which greedy $k$-means++ yields only an $\Omega(\ell \log k)$-approximation in expectation. More precisely, we analyze Algorithm~\ref{alg:greedykmeans++}.
\begin{algorithm}[h!]
	\caption{Greedy $k$-means++ algorithm~\cite{AV07}}\label{alg:greedykmeans++}
	\begin{algorithmic}[1]
		\State Sample a point $c_1$ independently and uniformly at random from $X$.
		\State Let $C=\{c_1\}$.
		\For{$i=2$ to $k$}
		  \State For all $x \in X$, set \[ p(x) := \frac{\min_{c \in C} ||x-c||^2}{\sum_{y\in X} \min_{c \in C} ||y-c||^2}. \]
		  \State Sample a set $S$ of $\ell$ points independently according to this probability distribution.
			\State Let $c_i=\arg \min_{u\in S} \Phi(X,C \cup \{u\})$. 
			\State Update $C=C\cup \{c_i\}$.
		\EndFor
		\State Run Lloyd's algorithm initialized with center set $C$ and output the result.
	\end{algorithmic}
\end{algorithm}

Note that we only draw one sample in the first step. This is due to the fact that in the first step, $k$-means++ is guaranteed to discover a new cluster, so there is no reason to draw multiple samples.

The instance is based on a regular $(k-1)$-simplex with side length $\sqrt{2}$. Let the vertices of this simplex be denoted by $v_1,\ldots,v_k$. There are $k$ points each at vertices $v_1,\ldots,v_{k-1}$, $(k-1)$ points at vertex $v_k$, and there is one point at the center~$o$ of the simplex. Let~$X$ denote the set of all these points. The simplex can be constructed explicitly in~$\RR^k$ by letting~$v_i$ be the $i$th canonical unit vector for each~$i$ and~$o=(1/k,\ldots,1/k)$. Then it follows that the distance between the center~$o$ and any vertex~$v_i$ is~$\sqrt{(k-1)/k}$.

An optimal clustering~$(\Cs{1},\ldots,\Cs{k})$ of this instance is obtained as follows: The clusters $\Cs{1},\ldots,\Cs{{k-1}}$ consist of the $k$ points at vertices $v_1,\ldots,v_{k-1}$, respectively, and the cluster $\Cs{k}$ consists of the $(k-1)$ points at vertex $v_k$ and the point at the center~$o$. The cost of this clustering is bounded from above by $||o-v_k||^2 = \frac{k-1}{k}=O(1)$.

\thmLowerBound*
\begin{proof}
Notice that for $\ell = 1$ there is nothing to show since a lower bound of $\Omega(\log k)$ is known for this case. So in the following, we assume that $\ell \ge 2$. Furthermore we consider first only the case that~$\ell \le \frac{k}{20\ln(k-1)}$ and defer the discussion of larger~$\ell$ to the end of the proof.

We consider the point set $X$ constructed above. 
Consider a $k$-clustering~$C$ obtained by greedy $k$-means++ that contains the point at~$o$ as one of the~$k$ centers. The costs of this clustering are at least~$(k-1)^2/k=\Omega(k)$ because there exists at least one~$i$ such that~$C$ has no center at~$v_i$. In the best case this is~$v_k$, which generates the aforementioned costs because the~$(k-1)$ points at~$v_k$ will be assigned to the center at~$o$. The approximation guarantee of this clustering is only~$\Omega(k)$. We will prove that with sufficiently large probability, greedy $k$-means++ places one of the centers at~$o$.

We start the analysis by defining the following events for all~$i\in[k]$:
\begin{align*}
   &F_i = \text{the center chosen in the $i$th iteration lies at~$v_k$},\\
   &G_i = \text{the center chosen in the $i$th iteration lies at~$o$},\\
   &H_i = F_i \cup G_i.
\end{align*}
We denote by~$\Phi_i$ the potential after $i-1$ iterations if in these iterations no point from~$\Cs{k}$ has been chosen as a center. Since the probability to choose the same $v_i$ more than once is zero, this means that 
 $i-1$ centers from different clusters from $\Cs{1},\ldots,\Cs{k-1}$ have been chosen. 
In the remaining $k-i+1$ clusters, $k$ points pay a cost of $2$, except for the one point at $o$ which pays $1-1/k$. Thus,
\[
   \Phi_i = 2((k-i+1)k-1) + 1-\frac{1}{k} 
\]
and
\[
    2((k-i+1)k-1) \le \Phi_i \le 2k(k-i+1).
\]

We define
\[
   F = F_1 \cup (F_2\cap\overline{H_1}) \cup \ldots \cup (F_{k-1}\cap\overline{H_1}\cap\ldots\cap\overline{H_{k-2}})
\]
as the event that in one the first~$k-1$ iterations a point at~$v_k$ is chosen as a center and that this is the first center chosen from~$\Cs{k}$. We exclude the last iteration because $\Pr(F_k)$ is significantly higher than $\Pr(F_{i})$ for $i \le k-1$.

We will prove a lower bound for the probability of the event~$\overline{F}\cap(G_2\cup\ldots\cup G_{k-1})$ because if this event happens then the point at~$o$ is one of the centers computed by greedy $k$-means++, i.e., the approximation factor is only~$\Omega(k)$. 

If the event~$F$ occurs then we cannot prove a lower bound on the approximation guarantee of greedy $k$-means++. Hence, we will prove an upper bound for the probability of~$F$. Observe that
\[
   \Pr[F] \le \sum_{i=1}^{k-1}\Pr[F_i\mid \overline{H_1}\cap\ldots\cap\overline{H_{i-1}}]\cdot \Pr[\overline{H_1}\cap\ldots\cap\overline{H_{i-1}}]
   \le \sum_{i=1}^{k-1}\Pr[F_i\mid \overline{H_1}\cap\ldots\cap\overline{H_{i-1}}]
\]
and
\[
   \Pr[F_1] = \frac{k-1}{k^2} \le \frac{1}{k}.
\]
Consider the situation that~$1 \le i-1\le k-2$ iterations have already been performed and that in these iterations cluster~$\Cs{k}$ has not been covered. Then each point from an uncovered cluster~$\Cs{j}$ with~$j<k$ reduces the potential by~$2k$. Each point at~$v_k$ reduces the potential by~$2(k-1)$ and the point at~$o$ reduces the potential by
\[
  (\underbrace{(k-i+1)}_{\ge 2}k-1)(1+1/k)+1-1/k \ge (2k-1)(1+1/k)+1-1/k = 2(k+1-1/k) > 2k.
\]
Hence, the points at~$v_k$ have the least potential reduction and thus a point at~$v_k$ is only selected as new center in iteration~$i$ if all~$\ell$ sampled candidates are at~$v_k$. Hence, we obtain 
\[
   \Pr[F_i\mid \overline{H_1}\cap\ldots\cap\overline{H_{i-1}}] = \left(\frac{2(k-1)}{\Phi_i}\right)^{\ell}.
\]
Altogether this implies
\begin{align*}
  \Pr[F]  \le \Pr[F_1] + \sum_{i=2}^{k-1}\Pr[F_i\mid \overline{H_1}\cap\ldots\cap\overline{H_{i-1}}]
          \le \frac{1}{k} + \sum_{i=2}^{k-1}\left(\frac{2(k-1)}{\Phi_i}\right)^{\ell}.
\end{align*}
Together with $\Phi_i \ge 2((k-i+1)k-1)$ this implies
\begin{align*}
 \Pr[F] & \le \frac{1}{k} + \sum_{i=2}^{k-1}\left(\frac{2(k-1)}{2((k-i+1)k-1)}\right)^{\ell}\\
        & = \frac{1}{k} + \sum_{i=2}^{k-1}\left(\frac{k-1}{(k-i+1)k-1}\right)^{\ell}\\
        & \le \frac{1}{k} + \sum_{i=2}^{k-1}\left(\frac{k}{(k-i+1)k}\right)^{\ell}\\
        & = \frac{1}{k} + \sum_{i=2}^{k-1}\left(\frac{1}{i}\right)^{\ell},
\end{align*}
where the inequality in the penultimate line of the calculation follows from~$\frac{a}{b} < \frac{a+1}{b+1}$ for~$0<a<b$. Using~$\ell\ge 2$ and $k\ge 4$, 
 it follows
\[
   \Pr[F] \le \frac{1}{k} + \sum_{i=2}^{k-1}\left(\frac{1}{i}\right)^{2} \le \frac{1}{k} + \sum_{i=2}^{\infty}\left(\frac{1}{i}\right)^{2} = \frac{1}{k}+\left(\frac{\pi^2}{6}-1\right) \le 0.9.
\]
This shows that with constant probability, the failure event~$F$ does not occur, i.e., with constant probability none of the points from~$v_k$ is chosen as a center in the first~$k-1$ iterations.

Now let us consider the probability that the point at~$o$ is selected as a center. We have argued above that the potential reduction of the point at~$o$ in iteration~$2\le i\le k-1$ is larger than~$2k$ if cluster~$\Cs{k}$ has not been covered in the first~$i-1$ iterations. We have also seen that any other point reduces the potential by at most~$2k$. Hence, in order to select the point at~$o$ as center it suffices already if it belongs to the $\ell$ candidates chosen in iteration~$i$. Denote the event that the $j$th sample in iteration $i$ is $o$ by $G_{ij}$. Then for~$i\in\{2,\ldots,k-1\}$,
\begin{align*}
   \Pr[G_i\mid\overline{H_1}\cap\ldots\cap\overline{H_{i-1}}] & = 
	\Pr[\cup_{j=1}^{\ell}G_{ij}\mid\overline{H_1}\cap\ldots\cap\overline{H_{i-1}}] \\
	& \ge \sum_{j=1}^{\ell} \Pr[G_{ij}\mid\overline{H_1}\cap\ldots\cap\overline{H_{i-1}}] - \sum_{1 \le j_1 < j_2 \le j} \hspace*{-0.5cm}\Pr[G_{ij_1}\cap G_{ij_2}\mid\overline{H_1}\cap\ldots\cap\overline{H_{i-1}}]\\
   & = \frac{\ell(1-1/k)}{\Phi_i} -  \frac{\binom{\ell}{2}(1-1/k)^2}{\Phi_i^2}\\
   & \ge \frac{\ell(1-1/k)}{\Phi_i} -  \frac{\ell^2(1-1/k)^2}{\Phi_i^2},
\end{align*}
where the first inequality follows from Bonferroni inequalities. 
Since~$\ell\le k/(20 \ln(k-1) )\le k/2$, we obtain
\begin{align*}
        \frac{\ell(1-1/k)}{\Phi_i} \le \frac{\ell}{\Phi_i} \le \frac{\ell}{2((k-i+1)k-1)} \le \frac{\ell}{k} \le \frac{1}{2}.
\end{align*}
This is helpful, because for any $a\in \mathbb{R}$ with $0 \le a \le 1/2$, it holds that $a-a^2 \ge a/2$. Thus, the previous two inequalities imply
\begin{equation}\label{eqn:PrGi}
   \Pr[G_i\mid\overline{H_1}\cap\ldots\cap\overline{H_{i-1}}] \ge 
	\frac{\ell(1-1/k)}{\Phi_i} - \left(\frac{\ell(1-1/k)}{\Phi_i}\right)^2
	\ge \frac{\ell(1-1/k)}{2\Phi_i}.
\end{equation}

Let us now condition on the event~$\overline{F}$, which happens with constant probability. Then we can write the probability of the event we care about as
\begin{align*}
  \Pr[\overline{F}\cap(G_2\cup\ldots\cup G_{k-1})] 
	& = \Pr[\overline{F}]\cdot\Pr[G_2\cup\ldots\cup G_{k-1}\mid\overline{F}] 
	  = \Pr[\overline{F}]\cdot\sum_{i=2}^{k-1}\Pr[G_{i}\mid \overline{F}]\\
	& \ge \Pr[\overline{F}]\cdot\sum_{i=2}^{k-1}\Pr[G_{i}\mid \overline{F_1}\cap \ldots \cap \overline{F_{i-1}}],
\end{align*}
where we used in the penultimate step that the events~$G_i$ are mutually exclusive and in the last step that $\overline{F} \subseteq \overline{F_1}\cap \ldots \cap \overline{F_{i-1}}$. We cannot use~\eqref{eqn:PrGi} directly to bound~$\Pr[G_{i}\mid \overline{F_1}\cap \ldots \cap \overline{F_{i-1}}]$ because the condition is different (in~\eqref{eqn:PrGi} we condition on the event that no point from~$\Cs{k}$ has been chosen as center in the first~$i-1$ iterations while conditioning on~$\overline{F_1}\cap \ldots \cap \overline{F_{i-1}}$ only implies that no point at~$v_k$ has been chosen as a center). 

To prove a lower bound on~$\Pr[G_2\cup\ldots\cup G_{k-1}\mid \overline{F_1}\cap \ldots \cap \overline{F_{i-1}}]$, we consider a different random experiment~$E$. This random experiment consists of~$k-2$ iterations numbered from~$2$ to~$k-1$ and each iteration~$i$ is successful with probability~$\Pr[G_i\mid\overline{H_1}\cap\ldots\cap\overline{H_{i-1}}]$ independent of the other iterations. Then~$\Pr[G_2\cup\ldots\cup G_{k-1}\mid \overline{F_1}\cap \ldots \cap \overline{F_{i-1}}]$ equals the probability that at least one of the iterations of~$E$ is successful. Let~$E'$ denote the same random experiment as~$E$ only with modified success probabilities. In~$E'$ iteration~$i$ is successful with probability~$\frac{\ell(1-1/k)}{2\Phi_i}$. Due to~\eqref{eqn:PrGi} and Bonferroni inequalities and using~$k\ge 4$, we obtain
\begin{align*}
   \Pr[G_2\cup\ldots\cup G_{k-1}\mid \overline{F}]
   &= \Pr[\text{at least one success in~$E$}]\\
   &\ge \Pr[\text{at least one success in~$E'$}]\\
   &\ge \sum_{i=2}^{k-1}\frac{\ell(1-1/k)}{2\Phi_i} - \sum_{2\le i<j\le k-1}\frac{\ell(1-1/k)}{2\Phi_i}\cdot \frac{\ell(1-1/k)}{2\Phi_j}\\
   &\ge \sum_{i=2}^{k-1}\frac{\ell(1-1/k)}{4k(k-i+1)} - \sum_{2\le i<j\le k-1}\frac{\ell}{4((k-i+1)k-1)}\cdot \frac{\ell}{4((k-j+1)k-1)}\\
   &\ge \sum_{i=2}^{k-1}\frac{\ell(1-1/k)}{4k(k-i+1)} - \sum_{2\le i<j\le k-1}\frac{\ell}{3k(k-i+1)}\cdot \frac{\ell}{3k(k-j+1)}\\
   & = \frac{\ell(1-1/k)}{4k} \sum_{i=2}^{k-1}\frac{1}{i} - \frac{\ell^2}{9k^2}\sum_{2\le i<j\le k-1}\frac{1}{(k-i+1)(k-j+1)}\\
   & \ge \frac{3\ell}{16k} \sum_{i=2}^{k-1}\frac{1}{i} - \frac{\ell^2}{9k^2}\left(\sum_{i=2}^{k-1}\frac{1}{i}\right)^2\\
    & \ge \frac{3\ell}{16k} (\ln(k-1)-1) - \frac{\ell^2}{9k^2}\ln^2(k-1).
\end{align*}
For~$k\ge 4$, we have $\ln(k-1)-1 \ge 0.089 \ln(k-1)$. Together with the previous calculation we get
\begin{align*}
  \Pr[G_2\cup\ldots\cup G_{k-1}\mid \overline{F}] & \ge 0.0166\cdot\frac{\ell\cdot \ln(k-1)}{k} - \left(\frac{\ell\cdot \ln(k-1)}{3k}\right)^2\\
  & = \frac{\ell\cdot \ln(k-1)}{k}\cdot\left(0.0166 - \frac{\ell\cdot \ln(k-1)}{9k}\right)\\
  & \ge 0.01\cdot \frac{\ell\cdot \ln(k-1)}{k},
\end{align*}
where we used~$\ell \le \frac{0.05\cdot k}{\ln(k-1)}$ for the last inequality.

Overall we obtain
\begin{align*}
  \Pr[\overline{F}\cap(G_2\cup\ldots\cup G_{k-1})]
  = \Pr[\overline{F}]\cdot\Pr[G_2\cup\ldots\cup G_{k-1}\mid\overline{F}]
  \ge 0.1\cdot 0.01\cdot \frac{\ell\cdot \ln(k-1)}{k} = \Omega\left(\frac{\ell\cdot\log(k)}{k}\right).
\end{align*}
If this event happens, then the costs of the clustering are~$\Omega(k)$. Hence the expected costs of the clustering computed by greedy $k$-means++ are $\Omega(\ell\cdot\log(k))$.

Finally let us consider the case~$\ell > \frac{k}{20\ln(k-1)}$. We argue that in this case the approximation guarantee cannot be better than for~$\ell=\frac{k}{20\ln(k-1)}$. To see that this is true, one has to have a closer look at where the upper bound on~$\ell$ has been used in the argument above. It is used twice: once for proving an upper bound on the conditional probability of~$G_i$ and once for proving an upper bound on the conditional probability of~$G_2\cup\ldots\cup G_{k-1}$. Both these probabilities increase with~$\ell$ so if~$\ell$ is larger one could simply replace it by~$\frac{k}{20\ln(k-1)}$, leading to a lower bound of $\Omega(k/\log(k) \cdot k)=\Omega(k)$ for the approximation guarantee.
\end{proof}


\section{Analysis of Noisy \texorpdfstring{$k$}{k}-means++ Seeding}

In this section we analyze a noisy seeding procedure, which we call \emph{noisy $k$-means++} in the following. This procedure iteratively selects~$k$ centers from the data set in a similar fashion as $k$-means++. The only difference is that the probability of sampling a point as the next center is no longer exactly proportional to its squared distance to the closest center chosen so far. The probabilities are only approximately correct. To be more precise, consider an iteration of noisy $k$-means++. For any point~$x\in X$, we denote by~$p_x$ the probability that~$x$ is chosen by $k$-means++ as the next center (i.e., $p$ is the uniform distribution in the first iteration and the distribution that results from $D^2$-sampling in the following iterations). In noisy $k$-means++ an adversary can choose an arbitrary probability distribution~$q$ on~$X$ with~$q_x \in [(1-\eps_1)p_x,(1+\eps_2)p_x]$ for all~$x\in X$ where $\eps_1\in [0,1)$ and $\eps_2\ge 0$ are parameters. Then the next center is sampled according to~$q$. This is repeated in every iteration of noisy $k$-means++ and in every iteration the adversary can decide arbitrarily how to choose~$q$ based on the current distribution~$p$ that results from $D^2$-sampling. We analyze the worst-case approximation guarantee provided by noisy $k$-means++.

Let us first introduce some notation. We denote by~$\Phi(X,C)$ the $k$-means costs of data set $X$ with respect to center set $C$, i.e.,
\[
   \Phi(X,C) = \sum_{x\in X} \min_{c\in C} ||x-c||^2.
\]
For~$c\in\RR^d$ we also write~$\Phi(X,c)$ instead of $\Phi(X,\{c\})$ and similarly for~$x\in\RR^d$ we write~$\Phi(x,C)$ instead of~$\Phi(\{x\},C)$. Let $\OPT_k(X)$ denote the optimal $k$-means costs of dataset $X$. In the following we assume that a data set~$X$ is given and we denote by~$(\Cs{1},\ldots,\Cs{k})$ an optimal $k$-clustering of~$X$. For a finite set~$X\subset\RR^d$, we denote by~$\mu(X)=\frac{1}{|X|}\sum_{x\in X}x$ its mean. The following lemma is well-known.
\begin{lemma}\label{lemma:kMeansCostsFormula}
For any finite $X\subset\RR^d$ and any $z\in\RR^d$,
\[
   \Phi(C,z) = \Phi(C,\mu(C)) + |C|\cdot ||z-\mu(C)||^2 = \OPT_1(C) + |C|\cdot ||z-\mu(C)||^2.
\]
\end{lemma}

We call an optimal cluster~$\Cs{i}$ \emph{covered} by (noisy) $k$-means++ if at least one point from~$\Cs{i}$ is selected as a center. Arthur and Vassilvitskii~\cite{AV07} observe that covered clusters are well approximated by $k$-means++ in expectation. In particular, they show that the expected costs of an optimal cluster~$\Cs{i}$ with respect to the center set computed by $k$-means++ are at most $2\cdot\OPT_1(\Cs{i})$ and $8\cdot\OPT_1(\Cs{i})$ if the cluster is covered in the first or any of the following iterations, respectively. First of all, we carry these observations over to noisy $k$-means++; the following two lemmata are straightforward adaptations of Lemma~3.2 and Lemma~3.3 in \cite{AV07}.

\begin{lemma}\label{lemma:ApproxUniform}
Let~$c_1$ denote the first center chosen by noisy $k$-means++. For each optimal cluster~$\Cs{i}$,
\[
    \E[\Phi(\Cs{i},c_1)\mid c_1\in\Cs{i}] \leq \frac{2(1+\eps_2)}{1-\eps_1}\cdot \OPT_1(\Cs{i}).
\]
\end{lemma}
\begin{proof}
In $k$-means++ the first center is chosen uniformly at random, i.e., each point from~$X$ has a probability of~$1/|X|$ of being chosen. In noisy $k$-means++, all points have a probability in~$[(1-\eps_1)/|X|,(1+\eps_2)/|X|]$ of being chosen. Hence, the probability of choosing a point~$x\in\Cs{i}$ as the first center conditioned on the first center being chosen from~$\Cs{i}$ is at most~$\frac{1+\eps_2}{(1-\eps_1)|\Cs{i}|}$. This implies
\begin{align*}
\E[\Phi(\Cs{i},\{c_1\})] &\leq \sum_{c\in \Cs{i}} \frac{1+\eps_2}{(1-\eps_1)|\Cs{i}|}~\Phi(\Cs{i},c)\\
&=\frac{1+\eps_2}{1-\eps_1}\cdot \frac{1}{|\Cs{i}|} \sum_{c\in \Cs{i}} \Phi(\Cs{i},c)\\
&=\frac{1+\eps_2}{1-\eps_1}\cdot \frac{1}{|\Cs{i}|} \sum_{c\in \Cs{i}} (\OPT_1(\Cs{i}) + |\Cs{i}| \cdot||c-\mu(\Cs{1})||^2)~~~~~\text{(Lemma~\ref{lemma:kMeansCostsFormula})} \\
&= \frac{2(1+\eps_2)}{1-\eps_1}\cdot \OPT_1(\Cs{i})\qedhere
\end{align*}
\end{proof}

\begin{lemma}\label{lemma:ApproxProportional} 
Consider an iteration of noisy $k$-means++ after the first one and let $C\neq\emptyset$ denote the current set of centers. We denote by~$z$ the center sampled in the considered iteration. Then for any~$C\neq\emptyset$ and any optimal cluster~$\Cs{i}$,
\[
  \E[\Phi(\Cs{i},C\cup \{z\})\mid C,z\in \Cs{i}]\leq \frac{8(1+\eps_2)}{1-\eps_1}\cdot \OPT_1(\Cs{i}). 
\]
\end{lemma}
\begin{proof} 
Conditioned on sampling a point from $\Cs{i}$, the probability of choosing point $x\in\Cs{i}$ as the next center is at most $\frac{1+\eps_2}{1-\eps_1}\cdot \frac{\Phi(x,C)}{\Phi(\Cs{i},C)}$. If~$x$ is chosen as the next center, the costs of any point $p\in\Cs{i}$ become $\min{(\Phi(p,C),||p-x||^2)}$. This implies
\begin{align}
\E[\Phi(\Cs{i},C\cup \{z\}) \mid C,z\in \Cs{i}] \notag
& = \sum_{x\in \Cs{i}} \Pr[z=x\mid C]\cdot\Phi(\Cs{i},C\cup \{x\})\notag\\
&\leq  \frac{1+\eps_2}{1-\eps_1}\cdot \sum_{x\in \Cs{i}}  \frac{\Phi(x,C)}{\Phi(\Cs{i},C)} \sum_{p\in \Cs{i}} \min{(\Phi(p,C),||p-x||^2)}.\label{eqn:PotentialCsi}
\end{align}
For any two points $x,p\in \Cs{i}$, we can write
\[
   \Phi(x,C) = \Big(\min_{c\in C}||x-c||\Big)^2 \le \Big(\min_{c\in C}(||x-p||+||p-c||)\Big)^2
   \le 2\Phi(p,C)+2||x-p||^2.
\]
By summing over all $p$ in $\Cs{i}$, we get 
\[
  \Phi(x,C)\leq \frac{2}{|\Cs{i}|} \sum_{p\in \Cs{i}} \Phi(p,C) + \frac{2}{|\Cs{i}|} \sum_{p \in \Cs{i}} ||x-p||^2.
\]
With~\eqref{eqn:PotentialCsi}, this implies that $\E[\Phi(\Cs{i},C\cup \{z\}) \mid C,z\in \Cs{i}]$ is bounded from above by
\begin{align*}
& \frac{1+\eps_2}{1-\eps_1}\cdot \sum_{x\in \Cs{i}}  \frac{\frac{2}{|\Cs{i}|} \sum_{p\in \Cs{i}} \Phi(p,C) + \frac{2}{|\Cs{i}|} \sum_{p \in \Cs{i}} ||x-p||^2}{\Phi(\Cs{i},C)} \sum_{p\in \Cs{i}} \min{(\Phi(p,C),||p-x||^2)}\\
&= \frac{1+\eps_2}{1-\eps_1}\cdot\sum_{z\in \Cs{i}} \frac{\frac{2}{|\Cs{i}|} \sum_{p\in \Cs{i}} \Phi(p,C)}{\sum_{p\in \Cs{i}} \Phi(p,C)} \sum_{p\in \Cs{i}} \min{(\Phi(p,C),||p-z||^2)} \\
&\quad+ \frac{1+\eps_2}{1-\eps_1}\cdot \sum_{z\in \Cs{i}} \frac{\frac{2}{|\Cs{i}|} \sum_{p \in \Cs{i}} ||p-z||^2}{\sum_{p\in \Cs{i}} \Phi(p,C)} \sum_{p\in \Cs{i}} \min{(\Phi(p,C),||p-z||^2)} \\
&\leq  \frac{1+\eps_2}{1-\eps_1}\cdot \sum_{z\in \Cs{i}} \frac{2}{|\Cs{i}|} \sum_{p\in \Cs{i}} ||p-z||^2 + \frac{1+\eps_2}{1-\eps_1}\cdot \sum_{z\in \Cs{i}} \frac{2}{|\Cs{i}|} \sum_{p \in \Cs{i}} ||p-z||^2 \\
&= \frac{4(1+\eps_2)}{1-\eps_1}\cdot \sum_{z\in \Cs{i}} \frac{1}{|\Cs{i}|} \sum_{p\in \Cs{i}} ||p-z||^2 \\
&= \frac{4(1+\eps_2)}{1-\eps_1}\cdot \sum_{z\in \Cs{i}} \frac{1}{|\Cs{i}|} (\OPT_1(\Cs{i})+|\Cs{i}|\cdot ||z-\mu(\Cs{i})||^2)~~~~~\text{(Lemma~\ref{lemma:kMeansCostsFormula})}\\
&= \frac{8(1+\eps_2)}{1-\eps_1}\cdot \OPT_1(\Cs{i})\qedhere
\end{align*}
\end{proof}

Consider a run of noisy $k$-means++. For~$t\in[k]$, let $H_t$ and $U_t$ denote the set of all points from~$X$ that belong after iteration~$i$ to covered and uncovered optimal clusters, respectively. Let~$u_t$ denote the number of uncovered clusters after iteration~$t$. Furthermore let~$T_t$ denote the set of centers chosen by noisy $k$-means++ in the first~$t$ iterations. We say that iteration~$t$ is \emph{wasted} if the center chosen in iteration~$t$ comes from~$H_{t-1}$, i.e., if in iteration~$t$ no uncovered cluster becomes covered.

\begin{cor}\label{cor:CostsCovered}
For any $t\in[k]$,
\[
  \Ex{\Phi(H_t,T_t)} \le \frac{8(1+\eps_2)}{1-\eps_1}\cdot \OPT_k(X)
\]
\end{cor}
\begin{proof}
Using Lemma~\ref{lemma:ApproxUniform} and Lemma~\ref{lemma:ApproxProportional} we obtain
\begin{align*}
  \Ex{\Phi(H_t,T_t)} & = \sum_{i=1}^k\Pr[\Cs{i}\subseteq H_t]\cdot \Ex{\Phi(\Cs{i},T_t)\mid \Cs{i}\subseteq H_t}\\
  & \le \sum_{i=1}^k\Pr[\Cs{i}\subseteq H_t]\cdot \frac{8(1+\eps_2)}{1-\eps_1}\cdot \OPT_1(\Cs{i})\\
  & \le \frac{8(1+\eps_2)}{1-\eps_1}\cdot\sum_{i=1}^k \OPT_1(\Cs{i})\\
  & = \frac{8(1+\eps_2)}{1-\eps_1}\cdot\OPT_k(X).\qedhere
\end{align*}
\end{proof}

Corollary~\ref{cor:CostsCovered} implies that the covered clusters contribute in expectation at most~$O(\OPT_k(X))$ to the costs of the solution computed by noisy $k$-means++ (assuming~$\eps_1$ and~$\eps_2$ to be constants). The not straightforward part is to prove an upper bound for the costs of the clusters that are not covered by noisy $k$-means++. For this, we adapt the analysis of $k$-means++ due to Dasgupta~\cite{dasgupta}. This analysis is based on a potential function that accumulates costs in every wasted iteration. The potential function has the properties that the expected value of the potential function in the end can be bounded and that the total costs accumulated are in expectation at least the costs of the uncovered clusters in the end. 

Dasgupta uses crucially that the expected average costs of the uncovered clusters do not increase in $k$-means++. For noisy $k$-means++ this is not true anymore in general. Hence, we have to adapt the potential function and the analysis. We define~$W_i=1$ if iteration~$i$ is wasted and~$W_i=0$ otherwise. Then the potential is defined as
\[
   \Psi_k = \sum_{i=2}^k W_i\cdot \frac{\Phi(U_{i},T_{i})}{u_{i}} .
\]

\begin{lemma}\label{lemma:PsiUpper}
It holds
\[
  \Ex{\Psi_k} \le \frac{8(1+\eps_2)^2}{(1-\eps_1)^2} \cdot(\ln(k)+1) \cdot \OPT_k(X).
\]
\end{lemma}
\begin{proof}
Let~$i\in\{2,\ldots,t\}$. In the following calculation we sum over all realizations~$\calF_{i-1}$ of the first $i-1$ iterations of noisy $k$-means++. Any realization~$\calF_{i-1}$ determines the value of~$\Phi(U_{i-1},T_{i-1})$ and~$u_{i-1}$. We use the notation~$[\ldots]_{\calF_{i-1}}$ to express that all terms inside the brackets take the values determined by~$\calF_{i-1}$. Then
\begin{align*}
  \Ex{W_i\cdot \frac{\Phi(U_{i},T_{i})}{u_{i}}}
  & = \sum_{\calF_{i-1}} \Pr[\calF_{i-1}] \cdot \Ex{W_i\cdot \frac{\Phi(U_{i},T_{i})}{u_{i}}\Bigm|\calF_{i-1}}\\
  & = \sum_{\calF_{i-1}} \Pr[\calF_{i-1}] \cdot \Pr[W_i=1\mid\calF_{i-1}]\cdot \Ex{\frac{\Phi(U_{i},T_{i})}{u_{i}}\Bigm|\calF_{i-1}\cap (W_i=1)}
\end{align*}
Since under the condition that iteration~$i$ is wasted the average costs of the uncovered clusters cannot increase, we can upper bound the term above by
\begin{align*}
  &  \sum_{\calF_{i-1}} \Pr[\calF_{i-1}] \cdot \Pr[W_i=1\mid\calF_{i-1}]\cdot \left[\frac{\Phi(U_{i-1},T_{i-1})}{u_{i-1}}\right]_{\calF_{i-1}} \\
  & \le \sum_{\calF_{i-1}} \Pr[\calF_{i-1}] \cdot \left[\frac{(1+\eps_2)\Phi(H_{i-1},T_{i-1})}{(1-\eps_1)\Phi(U_{i-1},T_{i-1})}\cdot \frac{\Phi(U_{i-1},T_{i-1})}{u_{i-1}}\right]_{\calF_{i-1}}\\
  & = \frac{1+\eps_2}{1-\eps_1}\cdot \sum_{\calF_{i-1}} \Pr[\calF_{i-1}] \cdot \left[\frac{\Phi(H_{i-1},T_{i-1})}{u_{i-1}}\right]_{\calF_{i-1}}\\
  & \le \frac{1+\eps_2}{1-\eps_1}\cdot \sum_{\calF_{i-1}} \Pr[\calF_{i-1}] \cdot \left[\frac{\Phi(H_{i-1},T_{i-1})}{k-i+1}\right]_{\calF_{i-1}}\\
  & = \frac{1+\eps_2}{1-\eps_1}\cdot \frac{\Ex{\Phi(H_{i-1},T_{i-1})}}{k-i+1}.
\end{align*}
This implies
\begin{align*}
   \Ex{\Psi_k} & = \sum_{i=2}^k \Ex{W_i\cdot \frac{\Phi(U_{i},T_{i})}{u_{i}}}\\
     & \le \frac{1+\eps_2}{1-\eps_1}\cdot \sum_{i=2}^k \frac{\Ex{\Phi(H_{i-1},T_{i-1})}}{k-i+1}.
\end{align*}
With Corollary~\ref{cor:CostsCovered} this yields
\[
  \Ex{\Psi_k} \le \frac{8(1+\eps_2)^2}{(1-\eps_1)^2}\cdot \OPT_k(X) \sum_{i=2}^k \frac{1}{k-i+1}
              \le \frac{8(1+\eps_2)^2}{(1-\eps_1)^2}\cdot (\ln(k)+1)\cdot \OPT_k(X).\qedhere
\]
\end{proof}

Observe that we increase the potential in each wasted iteration by the current average costs of the uncovered clusters. We have to make sure that in the end we have paid enough, namely at least the costs of the uncovered clusters remaining after $k$ iterations. If the average costs of the uncovered clusters do not increase then this is the case because in the end the number~$u_k$ of uncovered clusters equals the number of wasted iterations. Hence, if the average does not increase, then we have paid~$u_k$ times at least the final average~$\Phi(U_k,T_k)/u_k$ and thus in total at least~$\Phi(U_k,T_k)$. The following lemma, whose proof will be given further below, states that in expectation the average costs of the uncovered clusters can only increase by a logarithmic factor.

\begin{restatable}{lem}{lemIncAverage}\label{lemma:IncreaseAverage}
Set $D(\eps_1,\eps_2, \ln k) = \max\{18,\frac{24(\eps_1+\eps_2)(1+\eps_2)}{(1-\eps_1)^2}\}\cdot \ln k$.
If $k \ge D(\eps_1,\eps_2,\ln k)$, then set~$B(\eps_1,\eps_2,\ln k)=\frac{4(1+\eps_2)}{1-\eps_1} \cdot \ln(k)+2$, and otherwise, set $B(\eps_1,\eps_2,\ln k) = D(\eps_1,\eps_2,\ln k)$.
Then for any~$i\in[k]$ and any realization~$\calF_i$ of the first~$i$ iterations
\[
   \Ex{\frac{\Phi(U_k,T_k)}{u_k}\Bigm| \calF_i} \le B(\eps_1,\eps_2, \ln k) \cdot \left[\frac{\Phi(U_i,T_i)}{u_i}\right]_{\calF_i}.
\]
\end{restatable}

Next we relate the potential function to the costs of an optimal clustering.

\begin{lemma}\label{lemma:PsiLower}
Define $B(\eps_1,\eps_2,\ln k)$ as in Lemma~\ref{lemma:IncreaseAverage}.
Then
\[
  \Ex{\Psi_k} \ge \frac{\Ex{\Phi(U_k,T_k)}}{B(\eps_1,\eps_2,\ln k)}.
\]
\end{lemma}
\begin{proof}
For any~$i\in\{2,\ldots,k\}$, we obtain using Lemma~\ref{lemma:IncreaseAverage}
\begin{align*}
  \Ex{W_i\cdot \frac{\Phi(U_{i},T_i)}{u_{i}}}
  & = \sum_{\calF_{i},[W_i]_{\calF_i}=1} \Pr[\calF_i]\cdot\left[\frac{\Phi(U_{i},T_i)}{u_{i}}\right]_{\calF_{i}}\\
  & \ge \sum_{\calF_{i},[W_i]_{\calF_i}=1} \Pr[\calF_i]\cdot\frac{1}{B(\eps_1,\eps_2,\ln k)}\cdot \Ex{\frac{\Phi(U_k,T_k)}{u_k}\Bigm| \calF_i}\\
  & = \frac{1}{B(\eps_1,\eps_2,\ln k)}\cdot\sum_{\calF_{i}} \Pr[\calF_i]\cdot \Ex{W_i\cdot\frac{\Phi(U_k,T_k)}{u_k}\Bigm| \calF_i}\\
  & = \frac{1}{B(\eps_1,\eps_2,\ln k)}\cdot \Ex{W_i\cdot \frac{\Phi(U_k,T_k)}{u_k}}.
\end{align*}
Hence,
\begin{align*}
  \Ex{\Psi_k} & = \sum_{i=2}^k \Ex{W_i\cdot \frac{\Phi(U_{i},T_{i})}{u_{i}}}\\
  & \ge \frac{1}{B(\eps_1,\eps_2,\ln k)}\cdot\sum_{i=2}^k \Ex{W_i\cdot \frac{\Phi(U_k,T_k)}{u_k}}\\
  & = \frac{1}{B(\eps_1,\eps_2,\ln k)}\cdot\Ex{\left(\sum_{i=2}^k W_i\right)\cdot \frac{\Phi(U_k,T_k)}{u_{k}}}\\
  & = \frac{1}{B(\eps_1,\eps_2,\ln k)}\cdot\Ex{u_k\cdot \frac{\Phi(U_k,T_k)}{u_{k}}}\\
  & = \frac{\Ex{\Phi(U_k,T_k)}}{B(\eps_1,\eps_2,\ln k)}\qedhere
\end{align*}
\end{proof}

\thmNoisy*
\begin{proof}
For $k \ge \max\{18,\frac{24(\eps_1+\eps_2)(1+\eps_2)}{(1-\eps_1)^2}\}\cdot \ln k$, Lemma~\ref{lemma:PsiUpper} and Lemma~\ref{lemma:PsiLower} imply
\begin{align*}
  \Ex{\Phi(U_k,T_k)} & \le B(\eps_1,\eps_2,\ln k)\cdot \Ex{\Psi_k}\\
  & \le B(\eps_1,\eps_2,\ln k)\cdot \frac{8(1+\eps_2)^2}{(1-\eps_1)^2} \cdot(\ln(k)+1) \cdot \OPT_k(X)\\
  & = O\left(\frac{(1+\eps_2)^3}{(1-\eps_1)^3}\cdot \log^2(k)\cdot \OPT_k(X)\right).
\end{align*}
With Corollary~\ref{cor:CostsCovered} this implies
\begin{align*}
  \Ex{\Phi(H_k,T_k) + \Phi(U_k,T_k)} &\le 
  O\left(\frac{1+\eps_2}{1-\eps_1}\cdot \OPT_k(X)\right)
  + O\left(\frac{(1+\eps_2)^3}{(1-\eps_1)^3}\cdot \log^2(k)\cdot \OPT_k(X)\right)\\
  & \le O\left(\frac{(1+\eps_2)^3}{(1-\eps_1)^3}\cdot \log^2(k)\cdot \OPT_k(X)\right)
\end{align*}
For $k \le \max\{18,\frac{24(\eps_1+\eps_2)(1+\eps_2)}{(1-\eps_1)^2}\}\cdot \ln k$, we get
\begin{align*}
  &  B(\eps_1,\eps_2,\ln k)\cdot \frac{8(1+\eps_2)^2}{(1-\eps_1)^2} \cdot(\ln(k)+1) \cdot \OPT_k(X)\\
  & = O\left(\frac{(1+\eps_2)^4}{(1-\eps_1)^4}\cdot \log^2(k)\cdot \OPT_k(X)\right),
\end{align*}
where we use that $\eps_1 < 1$, so $\eps_1+\eps_2 \le 1+\eps_2$.
This implies
\begin{align*}
  \Ex{\Phi(H_k,T_k) + \Phi(U_k,T_k)} & \le O\left(\frac{(1+\eps_2)^4}{(1-\eps_1)^4}\cdot \log^2(k)\cdot \OPT_k(X)\right)\\
	& \le O\left(\frac{(1+\eps_2)^4}{(1-\eps_1)^4}\cdot \log^2 \left(\frac{(1+\eps_2)^2}{(1-\eps_1)^2}\right)\cdot \OPT_k(X)\right)\\
	& = O\left(\left(\frac{1+\eps_2}{1-\eps_1}\right)^4\cdot \log^2 \left(\frac{1+\eps_2}{1-\eps_1}\right) \cdot \OPT_k(X)\right),
\end{align*}
	where we use for the second inequality that either $\sqrt{k} \le \frac{k}{\ln k} \le 18$ and then $\ln k \le O(1)$, or
	\begin{align*}
	\sqrt{k} \le \frac{k}{\ln k} \le \frac{24(\eps_1+\eps_2)(1+\eps_2)}{(1-\eps_1)^2}
	& \Rightarrow \ln \sqrt{k} \le \ln \left(\frac{24(\eps_1+\eps_2)(1+\eps_2)}{(1-\eps_1)^2}\right)\\
	& \Rightarrow \log k \le O\left(\log \left(\frac{(1+\eps_2)^2}{(1-\eps_1)^2}\right)\right).\qedhere
	\end{align*}
\end{proof}

It remains to show that the average potential of uncovered clusters increases by at most a logarithmic multiplicative factor, i.e., to prove Lemma~\ref{lemma:IncreaseAverage}. We first consider the following abstract random experiment whose connection to noisy $k$-means++ we discuss in the actual proof of Lemma~\ref{lemma:IncreaseAverage} below. Let $a_1,\ldots,a_z\in \RR_{\ge 0}$ denote numbers with average value~$1$. Since there are $z$ numbers with the average equal to one, their sum equals~$z$. We assume that in each step of our experiment with probability~$\eps\in[0,1)$ an adversary chooses one of the numbers to be removed and with probability~$1-\eps$ a number is removed by proportional sampling (i.e., if number~$a_i$ still exists then it is removed with probability $a_i/S$, where~$S$ denotes the sum of the remaining numbers). Note that in this process the number $a_i$ is sampled with probability at least $(1-\eps)\frac{a_i}{S}$. Additionally after each step an adversary can arbitrarily lower the value of some numbers. This process is run for~$\ell$ steps and we are interested in an upper bound for the expected average of the numbers remaining after these $\ell$ steps. We denote this average by~$A_{\ell}$.

\begin{lemma}\label{lem:LargeL}
Let~$\eps\in(0,1)$, assume that~$\frac{z}{\ln z} \ge \max\{18,\frac{24\eps}{(1-\eps)^2}\}$, and~$\ell\ge z/2$. Then~$\Ex{A_{\ell}} \le \frac{4}{1-\eps} \cdot \ln(z)+2$.
For $z \le \max\{18,\frac{24\eps}{(1-\eps)^2}\} \cdot \ln z$, we observe that $\Ex{A_{\ell}} \le z \le \max\{18,\frac{24\eps}{(1-\eps)^2}\} \cdot \ln z$.
\end{lemma}
\begin{proof}
Let~$Z$ denote the number of adversarial steps among the first $\ell$ steps. Then~$\Ex{Z}=\eps \ell$. 
We denote by $\calF_1$ the event that~$Z \ge \frac{1+\eps}{2} \cdot \ell$. Note that $(1+\eps)/2 = \eps + (1-\eps)/2$ always lies between $\eps$ and $1$.

By Chernoff bound we get
\begin{align*}
   \Pr[\calF_1]  = \Pr\left[Z \ge \frac{1+\eps}{2}\ell\right] 
	  &= \Pr\left[Z \ge \frac{1+\eps}{2\eps}\eps \ell\right] \\
   & = \Pr\left[Z \ge \left(1+\frac{1-\eps}{2\eps}\right)\cdot\Ex{Z}\right]
    \le \exp\left(-\frac{\min\{\delta,\delta^2\}\cdot \Ex{Z}}{3}\right)
\end{align*}
for $\delta = \frac{1-\eps}{2\eps}$.
We make a case analysis. For $\frac{1-\eps}{2\eps} \ge 1 \Leftrightarrow \eps \le \frac{1}{3}$, $\min\{\delta,\delta^2\} = \delta$, so we have 
\[
\Pr[\calF_1] \le \exp\left(-\frac{\delta\cdot \Ex{Z}}{3}\right)
= \exp \left(-\frac{\frac{1-\eps}{2\eps}\cdot \eps \ell}{3}\right)
= \exp \left(-\frac{(1-\eps)}{6}\ell\right)
\le \exp \left(-\frac{(1-\eps)}{12}z\right),
\] 
where the last inequality follows from $\ell \ge z/2$. We observe that 
\[
\exp \left(-\frac{(1-\eps)}{12}z\right) \le \frac{1}{z} \Leftrightarrow \frac{z}{\ln z} \ge \frac{12}{1-\eps},
\]
	and by $\eps \le 1/3$ and by our lower bound on $z/\ln z$, we have $\frac{z}{\ln z} \ge 18 \ge \frac{12}{1-\veps}$.

If $\eps > 1/3$, we compute similarly that
\[
\Pr[\calF_1] \le \exp \left(-\frac{\frac{(1-\eps)^2}{(2\eps)^2}\cdot \eps \ell}{3}\right)
             \le \exp \left(-\frac{(1-\eps)^2}{12 \eps}\ell\right)
             \le \exp \left(-\frac{(1-\eps)^2}{24 \eps}z\right)
						\le \frac{1}{z}, \]
where the last inequality follows from $z/(\ln z) \ge \frac{24\eps}{(1-\eps)^2}$.

If the event~$\calF_1$ does not happen then in at least~$(1-\frac{1+\eps}{2})\ell=\frac{1-\eps}{2}\ell \ge \frac{1-\eps}{4}z =: z/c'$ steps proportional sampling is used to remove one of the numbers (we set $c' = 4/(1-\eps)$). 
We will show that with high probability after these steps all remaining numbers are at most~$2c'\ln{z}$. Let~$\calF_2$ denote the event that after $z/c'$ steps of proportional sampling at least one number with final value at least $2c'\ln{z}$ is remaining. Furthermore, let~$\calE_i$ denote the event that the $i$th number~$a_i$ remains after $z/c'$ steps of proportional sampling and its final value~$\tilde{a}_i$ is at least $2c'\ln{z}$ (remember that the adversary can decrease numbers during the process but not increase and hence~$\tilde{a}_i\le a_i$). Then~$\calF_2 = \calE_1\cup\ldots\cup \calE_z$. If~$\calE_i$ occurs then the $i$th number is in every step at least~$\tilde{a}_i\ge 2c'\ln{z}$. Since the numbers $a_1,\ldots a_z$ have average~$1$, their sum is~$z$. The sum of the remaining numbers cannot increase during the process. Hence, in every step the probability of taking the $i$th number is at least $(2c'\ln{z})/z$. This implies
\[
   \Pr[\calE_i] \le \left(1-\frac{2c'\ln{z}}{z}\right)^{z/c'} \le \exp\left(-2\ln{z}\right) = \frac{1}{z^2}.
\]
We use a union bound to obtain
\[
   \Pr[\calF_2] =
   \Pr[\exists i\in[z]: \calE_i] \le \frac{1}{z}.
\]

If neither~$\calF_1$ nor~$\calF_2$ occurs then the final value of each remaining number is at most~$2c'\ln{z}$. Hence, in this case, also the average is bounded from above by~$2c'\ln{z}$. Otherwise we only use the trivial upper bound of~$z$ for the average of the remaining numbers (observe that initially each~$a_i$ is at most~$z$ because the average is~$1$). Altogether we obtain
\begin{align*}
  \Ex{A_{\ell}} & \le \Pr[\neg\calF_1 \wedge \neg\calF_2]\cdot 2c'\ln{z} + \Pr[\calF_1\vee\calF_2]\cdot z\\
                & \le 2c'\ln{z} + (\Pr[\calF_1]+\Pr[\calF_2])\cdot z\\
                & \le 2c'\ln{z} + \left(\frac{1}{z}+\frac{1}{z}\right)\cdot z\\
                & = 2c'\ln(z)+2\\
				& = 4/(1-\eps) \ln z + 2.
\end{align*}
For the second inequality stated in the lemma, we only observe that even if we draw all but one number, the average cannot increase beyond $z$ since the sum of the numbers is $z$. Thus $\Ex{A_{\ell}} \le z$ is true for any $1 \le \ell \le z$. \qedhere
\end{proof}

\begin{lemma}\label{lem:SmallL}
Let~$\ell<z/2$. Then~$\Ex{A_{\ell}} \le 2$.
\end{lemma}
\begin{proof}
In the worst case all steps are adversarial and the $\ell$ smallest numbers are removed. Then the average of the remaining numbers is at most
\[
   \frac{z}{z-\ell} < \frac{z}{z-z/2} = 2.\qedhere
\]
\end{proof}

Together Lemma~\ref{lem:LargeL} and Lemma~\ref{lem:SmallL} imply the following corollary.
\begin{cor}\label{cor:RandomExperiment}
Let~$\eps\in(0,1)$ and $1 \le \ell \le z-1$.
Then for ${z} \ge \max\{18,\frac{24\eps}{(1-\eps)^2}\}\cdot {\ln z}$, we get
\[ 
  \Ex{A_{\ell}} \le \frac{4}{1-\eps}\cdot \ln z + 2, 
\]
and for $z \le \max\{18,\frac{24\eps}{(1-\eps)^2}\} \cdot \ln z$, we have $\Ex{A_{\ell}} \le \max\{18,\frac{24\eps}{(1-\eps)^2}\} \cdot \ln z$.
\end{cor}

Now we are ready to prove Lemma~\ref{lemma:IncreaseAverage}.

\lemIncAverage*
\begin{proof}
Given realization~$F_i$, after the first~$i$ iterations there are~$z=u_i\le k$ uncovered clusters. Each of them has certain costs with respect to the center set after the first~$i$ iterations. The costs of each cluster do not increase in the following iterations anymore because only new centers are added. In any iteration the costs of these clusters may decrease and one uncovered cluster may become covered. If the latter happens, the average costs of the uncovered clusters can increase (if the costs of the uncovered cluster that becomes covered are less than the average costs of the uncovered clusters). Hence, only the non-wasted iterations are of interest.

The costs of the uncovered clusters after the first~$i$ iterations correspond to the numbers~$a_1,\ldots,a_z$ in the random experiment above. We scaled the instance such that the sum of the~$a_i$ is equal to~$z$. This is without loss of generality. In each iteration of noisy $k$-means++ either a covered cluster is hit again, which can only reduce the numbers~$a_i$, or an uncovered cluster becomes covered, in which case the corresponding number is removed. Conditioned on covering an uncovered cluster, the probability~$p_i$ that~$a_i$ is removed is at least~$\frac{1-\eps_1}{1+\eps_2}\cdot\frac{a_i}{S}$, where~$S$ denotes the sum of the costs of the uncovered clusters (i.e., the sum of the remaining~$a_i$). We can simulate the probability distribution induced by the probabilities~$p_i$ by mixing two distributions: with probability $\frac{1-\eps_1}{1+\eps_2}$ we do proportional sampling, i.e., we choose~$a_i$ with probability~$\frac{a_i}{S}$, and with probability~$1-\frac{1-\eps_1}{1+\eps_2}$ we sample according to some other distribution to obtain the right probabilities~$p_i$. In the abstract random experiment analyzed above this second distribution is selected by an adversary. For~$\eps=\frac{\eps_1+\eps_2}{1+\eps_2}$ we have
\[
   1-\eps = 1-\frac{\eps_1+\eps_2}{1+\eps_2}
   = \frac{1+\eps_2-(\eps_1+\eps_2)}{1+\eps_2}   
   = \frac{1-\eps_1}{1+\eps_2}.
\]
Hence, Corollary~\ref{cor:RandomExperiment} applies to noisy $k$-means++ with~$\eps=\frac{\eps_1+\eps_2}{1+\eps_2}$.
Observe that then
\[
\frac{24 \eps}{(1-\eps)^2} = \frac{24 (\eps_1+\eps_2)(1+\eps_2)^2}{(1+\eps_2) (1-\eps_1)^2}
= \frac{24 (\eps_1+\eps_2) ( 1+\eps_2)}{(1-\eps_1)^2}.
\]
\end{proof}

\subsection{Applications}

A straightforward application of noisy $k$-means++ is the observation that probabilities are not computed exactly on actual machines. If we assume that the error can be bounded by constant multiplicative factors, then Theorem~\ref{thm:noisy} ensures that we still get at least an $\mathcal{O}(\log^2 k)$-approximation in expectation.

\paragraph{Not so greedy $k$-means++}

Furthermore, we consider the following variant of the greedy $k$-means++ algorithm (Algorithm~\ref{notsogreedykmeans++}).

\begin{algorithm}[hb]
	\caption{Moderately greedy $k$-means++}\label{notsogreedykmeans++}
	\begin{algorithmic}[1]
		\State {\bf Input}: Set $X\subseteq \R^d$, integers $k,l$
		\State {\bf Output}: $C\subseteq X, |C|=k$
		\State $C=\emptyset$
		\State Sample a point $c_1$ independently and uniformly at random from $X$.
		\State Let $C=\{c_1\}$.
		\For{$i=2$ to $k$}
		  \State With probability $p$, sample one point $c_i$ with $D^2$-sampling and set $C=C\cup \{c_i\}$.
			\State With the remaining probability:
      \Statex \quad\quad\quad Sample a set $S$ of $\ell$ points independently with $D^2$-sampling from $X$ with respect to $C$.
			\Statex \quad\quad\quad Let $c_i=\arg \min_{u\in S} \Phi(X,C \cup \{u\})$. 
			\Statex \quad\quad\quad Update $C=C\cup \{c_i\}$.
		\EndFor
		\State Return $C$
	\end{algorithmic}
\end{algorithm}

Let $x \in P$ be any point. Say that $p_i(x)$ is the probability to draw $x$ with one $D^2$-sample from $X$ based on the center set $c_1,\ldots,c_{i-1}$. Then the probability $q_i(x)$ to sample $x$ in iteration $i$ of the above algorithm satisfies
\[
p \cdot p_i(x) \le q_i(x) \le [(1-p) \cdot \ell + p]\cdot p_i(x),
\]
since with probability $p$, we do exactly the same as $k$-means++, and with probability $(1-p)$, we sample $\ell$ times, which can at most boost the probability by a factor of $(1-p) \cdot \ell$.
Assume that $p$ is a constant. Then by Theorem~\ref{thm:noisy}, moderately greedy $k$-means++ has an expected approximation guarantee of $\mathcal{O}(\ell^3\cdot \log^2 k)$ (for large $k$).

\section{Conclusions and Open Problems}

We have initiated the theoretical study of greedy $k$-means++. Our lower bound of~$\Omega(\ell\cdot\log{k})$ demonstrates that in the worst case the best choice is~$\ell=1$, i.e., the normal $k$-means++ algorithm. One question that arises immediately is if there exists a matching upper bound. Since greedy $k$-means++ is successful in experiments, it would also be interesting to find a notion of clusterability under which it outperforms $k$-means++. Our lower bound example, however, satisfies many of the common clusterability assumptions such as separability \cite{ORSS2012} and approximation stability \cite{BBV2009}.

It would be interesting to see if our analysis of noisy $k$-means++ can be improved. In particular, we do not believe that the additional $\mathcal{O}(\log{k})$ term coming from Corollary~\ref{cor:RandomExperiment} is necessary. We conjecture that in this corollary~$\Ex{A_{\ell}}$ is bounded from above by a constant depending on~$\epsilon$ for every~$\ell$.

\section{Acknowledgements}
The authors thank Emre Celebi for pointing that greedy $k$-means++ is discussed in~\cite{V07}.

\bibliographystyle{plain}
\bibliography{ref}
\end{document}